\documentclass[UKenglish]{lipics}

\usepackage{color}
\usepackage{graphicx}
\newtheorem{observation}{Observation}
\newcommand{\myparagraph}[1]{\medskip\noindent{\bf #1}}

\newcommand{\leaveout}[1]{}

\renewcommand{\medskip}{\smallskip}
\renewcommand{\int}{{\ensuremath{\rm int\,}}}

\usepackage{algorithm}
\newcommand{\IF}[1]{{\bf if} $\left(\mbox{#1}\right)\quad$}
\newcommand{\WHILE}[1]{{\bf while} $\left(\mbox{#1}\right)\quad$}
\newcommand{\RETURN}{{\bf return} }

\newcommand{\AND}{{\bf and} }
\newcommand{\LOOP}{{\bf loop}}
\newcommand{\ENDLOOP}{{\bf end loop}}

\date{}
\title{Optimum-width upward drawings of trees} 
\authorrunning{T. Biedl}

\author[1]{Therese~Biedl}

\affil[1]{
David R. Cheriton School of Computer Science, University of Waterloo, Canada.
\texttt{biedl@uwaterloo.ca}}

\Copyright{Biedl}

\subjclass{I.3.5 Computational Geometry and Object Modeling}

\keywords{tree drawing, upward, order-preserving, optimum width}

\serieslogo{}
\volumeinfo
  {Billy Editor, Bill Editors}
  {2}
  {Conference title on which this volume is based on}
  {1}
  {1}
  {1}
\EventShortName{}
\DOI{10.4230/LIPIcs.xxx.yyy.p}

\begin{document}
\maketitle
\begin{abstract}
An {\em upward} drawing of a tree is a drawing 
such that no parents are below their children.  It is {\em order-preserving}
if the edges to children appear in prescribed order around each node.
Chan showed that any tree has an upward order-preserving drawing with width
$O(\log n)$.  In this paper, we present linear-time algorithms
that finds upward with {\em instance-optimal} 
width, i.e., the width is the minimum-possible for the input tree.    

We study two different models.  In the first model, the drawings need not
be order-preserving; a very simple algorithm then finds straight-line
drawings of optimal width.  In the second model, the drawings must be
order-preserving; and we give an algorithm that finds optimum-width
{\em poly-line drawings}, i.e., edges are allowed to 
have bends.
We also briefly study order-preserving upward {\em straight-line} drawings,
and show that some trees require larger width if drawings must 
be straight-line.
\end{abstract}

\section{Introduction}

An {\em ideal drawing} of a tree \cite{Chan02} is one that is planar
(no edges cross), strictly-upward (the curves from parents to children
are strictly $y$-monotone), order-preserving (a given order of 
children is maintained in the drawing) and straight-line (edges are drawn 
as straight-line segments).  
For such drawings, the height must be at least the (graph-theoretic) height 
of the tree, and hence to achieve a small area  one focuses
on finding a small width.
Chan \cite{Chan02} gave algorithms that achieve 
ideal drawings of area $O(n4^{\sqrt{2\log n}})$ and width 
$O(2^{O(\sqrt{\log n)}})$.   He also briefly mentioned that
a variant of the algorithm achieves width $O(\log n)$, and 
one can additionally achieve height $O(n)$ by adding one bend per edge.%
\footnote{Di Battista and Frati \cite{DF14} asked later whether trees
have upward order-preserving poly-line drawings of area $O(n\log n)$; Chan's remark 
proves this.}
For binary trees, Garg and Rusu 
showed that $O(\log n)$ width and $O(n\log n)$ area can be achieved even
for straight-line drawings \cite{GR03}.
See the recent overview paper by Frati and Di Battista \cite{DF14} for
many other related results.

\myparagraph{Our results: }
This paper was motivated by the quest of finding ideal drawings for which
the width is {\em instance-optimal}, i.e., tree $T$ is 
drawn with the smallest width that is possible for $T$.     This problem
remains unsolved.  We here relax the restrictions in two ways.  In the
first relaxation, we drop ``order-preserving''.  Here a very simple 
modification of a known algorithm gives strictly-upward straight-line planar 
drawings of
instance-optimal width.  
(For the rest of this paper,
all drawings are required to be planar, and we will sometimes omit this 
quantifier.)

Secondly, for the main result of our paper,
we drop ``straight-line''
and study poly-line drawings, i.e., edges may have bends.
We give a linear-time algorithm to find order-preserving strictly-upward 
poly-line drawings of trees that have optimal width.
Our construction produces strictly-upward drawings, but the argument that
this is optimal works also for {\em upward} drawings (where edge-segments may be
horizontal).  In particular therefore, the optimum width is the same for
upward and strictly-upward order-preserving poly-line drawings.  As another side-effect,
we show that the root can always be required to be at the top left or the
top right corner without increasing width.
We also briefly discuss straight-line drawings, and show
that these sometimes require a larger width than poly-line drawings.

Phrasing our results in terms of $n$, we can show that the grid-size of
our drawings is never more than $\log(n+1)\times n$ for unordered drawings,
and not more than $(\log n + 1)\times (2n-1)$ for order-preserving poly-line drawings.
In particular this gives another independent proof that trees have
order-preserving poly-line drawings with area $O(n\log n)$.

\myparagraph{Related results: }
To our knowledge no previous paper addressed the issue of finding 
upward tree drawings with instance-optimal width.  
Alam et al.~\cite{ASR+10} showed how to find upward tree drawings with 
instance-optimal height, both in the order-preserving and the unordered model.
If we drop the ``upward'' restriction, then testing whether a planar graph
can be drawn such that one dimension (usually chosen to be the height)
is at most $k$ is fixed-parameter tractable in $k$ \cite{DFK+08}.
Algorithms to minimize this smaller dimension are known for trees
\cite{MAR11} and approximation algorithms for this smaller dimension are known
for trees \cite{Sud04}, outer-planar graphs \cite{Bie-WAOA12}, and
Halin-graphs \cite{Bie-Halin}.

\myparagraph{A few notations: }
Let $T$ be a tree with $n$ nodes rooted at node $u_r$. 
Let $c_1,\dots,c_d$ be the children of the root,
where $d=\deg(u_r)$ is the {\em degree} of $u_r$.
For any child $c_i$, let $T_{c_i}$
be the sub-tree rooted at child $c_i$.
If the tree is ordered, then
we assume that the children are enumerated from left to right, and 
we say that $c_i$ is ``left of $c_j$''
if $i\leq j$, and ``strictly left of $c_j$'' if $i<j$.
Similarly define ``right of'', ``strictly right of'', ``between''
and ``strictly between''.  

We aim to find a poly-line drawing of $T$, which means that
every edge is represented by a {\em poly-line}, i.e., a piecewise
linear curve.  In a {\em straight-line drawing}, edge curves have
no bends.
All drawings in this paper require that nodes and bends 
of poly-lines have an
integral $x$-coordinate.  The {\em width} of such a drawing is
the smallest $W$ such that (after possible translation) all
$x$-coordinates are between $1$ and $W$.  {\em Column $X$}
describes the vertical line with $x$-coordinate $X$.  In some
situations we analyze the height as well, and then require that
all nodes and bends have an integral $y$-coordinate and measure
the height by the number of rows intersected by the drawing.

\section{Optimum-width unordered straight-line drawings}
\label{sec:draw}
\label{sec:unordered}

We first briefly consider unordered drawings, and show here that a
simple algorithm achieves optimum width.  The key idea is to express
this optimum width as a different graph-parameter that is easily 
computed.

\begin{definition}
\label{def:rpw}
The {\em rooted pathwidth of $T$} (denoted $rpw(T)$) is defined as follows:  
$$
rpw(T) = \left\{
\begin{minipage}{70mm}
$
\begin{array}{ll}
1 & \mbox{if $T$ is a single node} \\
\min_{c_h} \max_{c} \{ rpw(T_c) + \chi_{(c\neq c_h)}\} & \mbox{otherwise}
\end{array}
$
\end{minipage}
\right.
$$
Here the minimum is taken over all possible choices of one child $c_h$
of the root, the maximum is taken over all possible choices of children $c$
of the root, and $\chi$ denotes the {\em characteristic function},
i.e., $\chi_{(c\neq c_h)}$ is 1 if $c\neq c_h$ and 0 otherwise.  
A child $c_h$ where the minimum is achieved
is called the {\em rpw-heaviest} child (breaking ties arbitrarily).
\end{definition}

The rooted pathwidth can be computed 
in linear time using a bottom-up approach.
For some arguments it helps to know an equivalent definition of
rooted pathwidth.
A {\em root-to-leaf path} in $T$ is any path 
in $T$ that connects the root to
one of the {\em leaves}, i.e., one of the nodes that have no children.
We call $T$ a {\em rooted path} if $T$ is a path from the root to 
a (unique) leaf.  One can easily show the following (see the appendix for details):

\begin{observation}
\label{obs:rpw}
We have
$rpw(T)=1$ if $T$ is a rooted path, and $rpw(T)=\allowbreak
\min_{P}\allowbreak \max_{T' \subset T-P} \allowbreak \left\{ 1+ rpw(T')\right\}$ otherwise.
Here, the minimum is taken over all root-to-leaf paths $P$, and the maximum
is taken over all subtrees $T'$ of $T-P$.
%
\end{observation}

\myparagraph{Example: } Consider the tree in Fig.~\ref{fig:ex}(a).
The numbers denote the rooted pathwidth of the subtree, computed with 
the formula in Definition~\ref{def:rpw}.
If we remove the root-to-leaf path $P$, then all subtrees of $T-P$
are singletons or rooted paths, and hence have rooted pathwidth 1.
Therefore $rpw(T)\leq 2$ if we use the formula of Observation~\ref{obs:rpw}.

\begin{figure}[ht]
\centering
\begin{tabular}{ccccc}
\includegraphics[page=1,width=30mm]{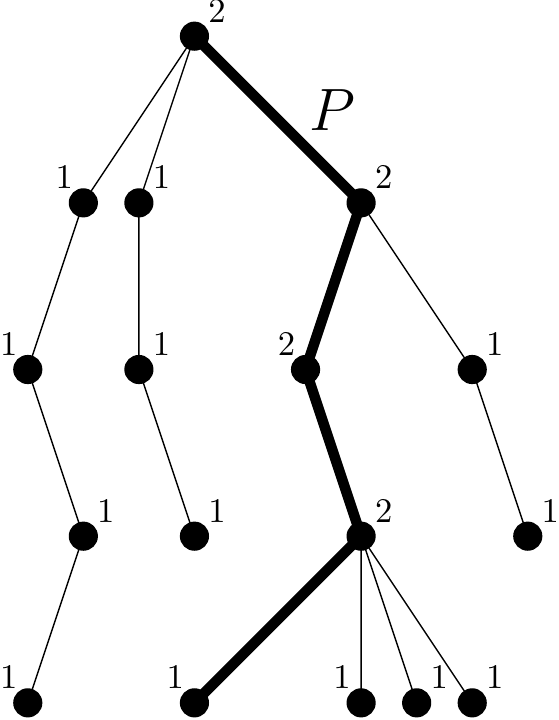} &
\hspace*{10mm} &
\includegraphics[width=15mm]{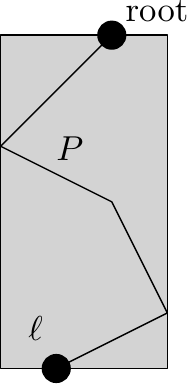} &
\hspace*{10mm} &
\includegraphics[width=20mm]{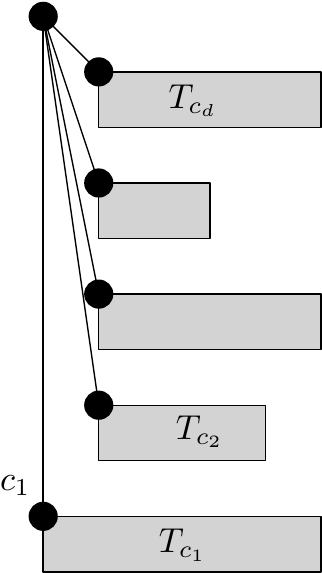} \\
(a) && (b) && (c) 
\end{tabular}
\caption{(a) Example.  (b) Lower bound.  (c) ``Standard'' construction.}
\label{fig:ex}
\label{fig:tree_lower}
\label{fig:tree_draw}
\end{figure}

The name ``rooted pathwidth'' was chosen because the rooted
pathwidth is closed related to the graph parameter ``pathwidth $pw(T)$'' 
of a tree (see e.g.~\cite{Sud04}).  One can easily show that
$pw(T)\leq rpw(T)\leq 2pw(T)+1$ for any rooted tree; 
see the appendix.
Now we show the relationship between rooted pathwidth and
width of drawings.
Note that the following lower bound even holds for the weaker models of upward
(vs. strictly-upward) and poly-line (vs. straight-line) drawing, while the
upper bound yields a construction in the strongest model.

\begin{lemma}
\label{lem:lower}
Let $\Gamma$ be any upward poly-line drawing of a rooted tree $T$.  Then the
width $W$ of $\Gamma$ is at least $rpw(T)$.
\end{lemma}
\begin{proof}
Since $\Gamma$ is an upward drawing, the root of $T$ has the 
maximal $y$-coordinate.  Let $\ell$ be the leaf that has the minimal 
$y$-coordinate in $\Gamma$, breaking ties arbitrarily.    
Since $\Gamma$ is an upward drawing, no other node can have smaller
$y$-coordinate than $\ell$.  Let $P$ be the unique path from the
root to $\ell$ in $T$.

If $T=P$, then $T$ is a rooted path and so $rpw(T)=1\leq W$.  Else
consider any rooted subtree $T'$ of $T-P$.  
The drawing $\Gamma'$ of $T'$ induced by
$\Gamma$ must have width at most $W-1$, because
path $P$ connects the topmost with the bottommost row in $\Gamma$,
and hence any connected component of $\Gamma-P$ intersects at
most $W-1$ columns.
By induction, therefore $rpw(T')\leq W-1$ for all subtrees
$T'$ of $T-P$, and so $rpw(T)\leq W$.
\end{proof}

\begin{lemma}
Any rooted tree $T$ has a strictly-upward straight-line drawing of width 
at most $rpw(T)$.  Moreover, the root is drawn in the top-left corner.
\label{lem:upper}
\end{lemma}
\begin{proof}
Such a drawing can be found by modifying
the algorithm of Crescenzi et al.~\cite{CDP92}.  The claim is trivial
if $T$ is a single node.  So assume $T$ has 
children $c_1,\dots,c_d$ and for $i=1,\dots,d$ draw
$T_{c_i}$ recursively with width $rpw(T_{c_i})$. After possible reordering
of children we may  assume 
that $c_1$ is the rpw-heaviest child, which implies that
$rpw(T_{c_i})<rpw(T)$ for all $i>1$.
Place the drawings of $T_{c_d},T_{c_{d-1}},\dots,T_{c_2},T_{c_1}$, 
one above the other, such that the root of $T_{c_i}$ is in column 2 for 
$i=d,\dots,2$ and in column 1 for $i=1$.  See Fig.~\ref{fig:tree_draw}(c).
Clearly we can connect $v$ to all its
children without crossing and the width is
$\max\{ rpw(T_{c_1}), \max_{i>1} \{ rpw(T_{c_i})+1 \} \}$,
which is at most $rpw(T)$ by choice of $c_1$.
\end{proof}

Observe that the height of the drawing is $n$, since every row intersects
exactly one node.  The width is no more than $\log(n+1)$ (see the appendix).
Since the rooted pathwidth (and with it the rpw-heaviest child for each
node) can be found in linear time, we therefore have:

\begin{theorem}
There exists a linear-time algorithm to create for any rooted tree $T$ a
planar strictly-upward straight-line drawing of optimal width and height $n$.
\end{theorem}


\section{The rank-function}

Now we turn towards {\em order-preserving} drawings of tree, so assume
from now on that for every node the children have a fixed order.  We
will find poly-line drawings that have the minimum-possible width.
The key idea is again to express the optimum width of 
a drawing of tree $T$ via a recursive function that depends solely on the
structure of the tree.    However, this function (which we call the
{\em rank}) is significantly more complicated than the rooted
pathwidth.

\begin{definition}
Let $T$ be a tree  and let $c_1,\dots,c_d$ be the children of the
root from left to right.  Define the {\em rank} $R(T)$ to be $1$ if $T$ is a single-node tree, and
to be the smallest value $W$ such that there exists a rank-$W$-witness
for $T$ otherwise.  Here, for a given integer $W\geq 1$, a 
{\em rank-$W$-witness} for $T$ consists of the following:
\begin{itemize}
\item a {\em classification} of each child as either {\em big} or {\em small},
\item a {\em coordinate} $X$,  i.e., an integer with $1\leq X\leq W$, and
\item an {\em index of the vertical child}, i.e., an index $v\in \{1,\dots,d\}$ such  that $c_v$ is a big child.
\end{itemize}
Such a rank-$W$-witness must satisfy the following {\em rank-conditions}:
\begin{description}
\item[(R1$\ell$)] At most $X-1$ big children are strictly left
	of $c_v$.
\item[(R1r)] At most $W-X$ big children are strictly right of $c_v$.
\item[(R2$\ell$)]  Any small child $c_i$
	with $i<v$ satisfies $R(T_{c_i})\leq X-1-\ell_i$, where
	$\ell_i$ is the number of big children to the left of $c_i$.
\item[(R2r)]  Any small child $c_i$
	with $i>v$ satisfies $R(T_{c_i})\leq W-X-r_i$, where
	$r_i$ is the number of big children to the right of $c_i$.
\item[(R3)]  The ranks of the big children are dominated by a
	permutation of $\{1,\dots,W\}$.  In other words,
	one can assign a {\em rank-bound} $\pi(c_i)\in \{1,\dots,W\}$ to each
	big child $c_i$
	such that $R(T_{c_i})\leq \pi(c_i)$ and $\pi(c_i)\neq \pi(c_j)$ for
	$c_i\neq c_j$.
\end{description}
\end{definition}

Fig.~\ref{fig:witness}(left) illustrates this concept.
For ease of wording, we often say 
``the rank of $c_i$'' in place of ``the rank of the tree rooted at $c_i$''.
To explain the naming for rank-$W$-witnesses: we will later see
that there exists a drawing that has width $W$, value $X$ 
is the $x$-coordinate of the root, the big children are
those children where the drawing of the subtree intersects column $X$,
and the vertical child is the child for which the edge leaves the root
vertically.  The following easy result will be needed later:

\begin{observation}
\label{obs:W}
If a tree has rank $W\geq 2$, then all children of the root have rank at most $W$,
and at most one child has rank exactly $W$.
\end{observation}
\begin{proof}
Fix an arbitrary rank-$W$-witness.  By (R3) there are rank-bounds, which means
that all big children have rank at most $W$ and
at most big one child has rank equal to $W$.
By (R2$\ell$) and (R2r), any small child has rank at most $\max\{X-1,W-X\}$,
and by $1\leq X\leq W$ this is at most $W-1$.
\end{proof}

We also use a special type of witness, which we will later see to correspond
to a rank-$W$-witness with $X=1$ and $v=1$.

\begin{definition}
Let $T$ be a tree with $n\geq 2$ nodes and let $c_1,\dots,c_d$ be
the children of the root from left to right.
For $W\geq 2$, a {\em left-corner-$W$-witness} of $T$ consists
of a number $1\leq W'\leq W+1$ and
a sequence $\sigma(W') <\dots < \sigma(W)$
such that: 
\begin{description}
\item[(C1)] $T_{c_{\sigma(w)}}$ has rank $w$ for all $w\in \{W',\dots,W\}$
\item[(C2)] For any $i$ with $\sigma(w)<i<\sigma(w+1)$, $T_{c_i}$
	has rank at most $w-1$.  Here $w\in \{W'-1,\dots,W\}$, and we define
	$\sigma(W'-1):=0$ and $\sigma(W+1):=d+1$.
\end{description}

Symmetrically, a {\em right-corner-$W$-witness} consists of a number $1\leq W'\leq W+1$ 
and a sequence $\sigma(W) > \dots, > \sigma(W')$
such that for all $w\in \{W',\dots,W\}$ child $c_{\sigma(w)}$
has rank $w$, and the children strictly between $c_{\sigma(w+1)}$ and
$c_{\sigma(w)}$ have rank at most $w-1$.  A {\em corner-$W$-witness}
is a left-corner-$W$-witness or a right-corner-$W$-witness.
\end{definition}

\begin{figure}[ht]
\hspace*{\fill}
\includegraphics[height=40mm,page=1]{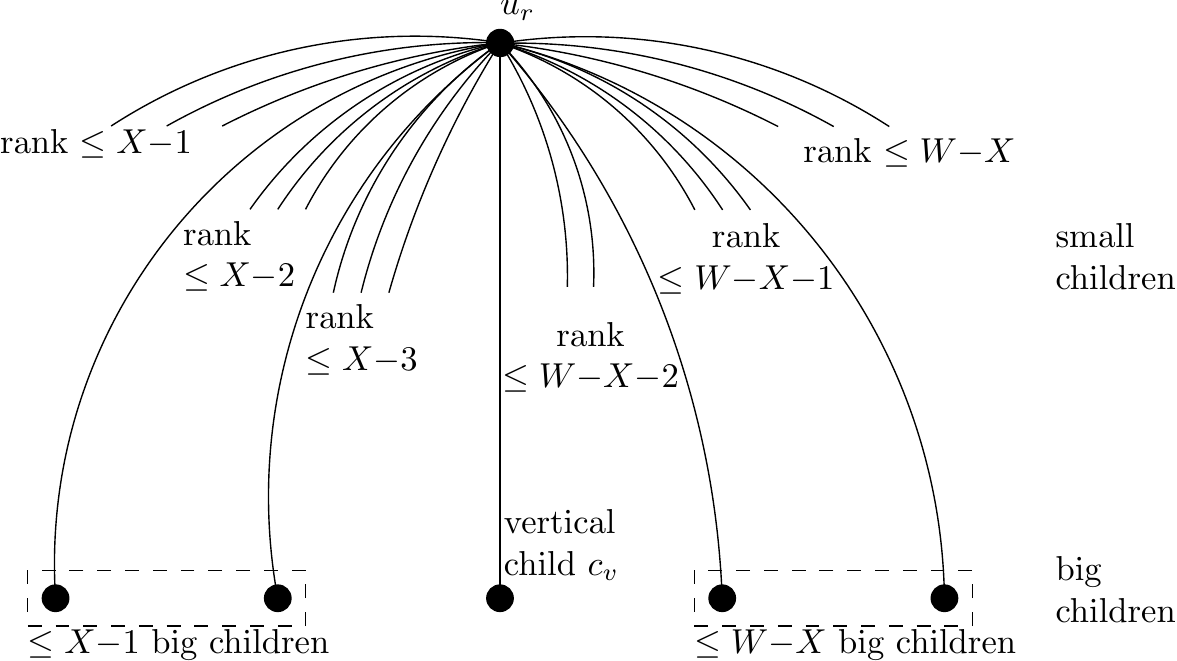} 
\hspace*{\fill}
\includegraphics[height=40mm,page=2,trim = 30 0 30 0, clip]{witnessIllustration.pdf}
\hspace*{\fill}
\caption{Illustration for (left) a 
rank-$W$-witness and (right) a left-corner-$W$-witness.}
\label{fig:witness}
\end{figure}

Notice that the definition of left-corner-$W$-witness specifically allows $W'=W+1$;
in this case no $\sigma(\cdot)$ needs to be given, (C1) is vacuously true, and (C2)
holds if and only if all children have rank at most $W-1$.  In particular this
shows:

\begin{observation}
\label{obs:W+1}
Let $T$ be a tree with $n\geq 2$ nodes, and assume all children have rank
at most $W-1$.  Then $T$ has a left-corner-$W$-witness.
\end{observation}

\myparagraph{Outline: } 
We briefly outline our approach to finding optimum-width 
poly-line drawings.
First, we show in Section~\ref{sec:corner2draw} that from a 
left-corner-$W$-witness, we can easily construct a drawing
of width $W$.  A symmetric construction converts
a right-corner-$W$-witness into a drawing of width $W$.  
Next, we show in Section~\ref{sec:draw2rank}
that from any (planar, upward, order-preserving) drawing of width $W$
we can extract a rank-$W$-witness.  Finally, to close the
cycle, we show in Section~\ref{sec:rank2corner} that any
rank-$W$-witness implies the existence of a corner-$W$-witness.
Hence the rank of a tree
equals the minimum width of an upward order-preserving drawing.
The proof in Section~\ref{sec:rank2corner} is constructive
and in particular allows to test in linear time whether 
a corner-$W$-witness exists.
Since the construction in Section~\ref{sec:corner2draw}
also takes linear time, this shows the following:

\begin{theorem}
\label{thm:main}
For any tree $T$, we can find  in linear time a planar strictly-upward 
order-preserving poly-line drawing that has optimum width. 

Moreover, the root is placed at the top-left or top-right corner,
and we can either choose to have linear height and at most 3 bends per edge,
or to have at most 1 bend per edge.
\end{theorem}

We find it especially interesting that 
we can always assume the root to be at a corner without increasing width.
Many previous tree-drawing algorithms 
(e.g.~\cite{Chan02,CDP92,GR03}) created drawings with the root
at a corner, but
proving, without going through rank-witnesses, that the
root can be moved to a corner without increasing width seems daunting.
Indeed, as we show in 
Section~\ref{sec:straight}, this claim is not true for straight-line
drawings.

\section{From rank-witness to drawing}
\label{sec:corner2draw}

To create drawings using rank-witnesses, we need a
result whose lengthy proof is deferred to 
Section~\ref{sec:rank2corner}:

\begin{lemma}
\label{lem:rank2corner}
Any $T$ with rank $W$ has a corner-$W$-witness.
\end{lemma} 

\begin{lemma}
\label{lem:corner2draw}
Any $n$-node tree $T$
has a planar strictly-upward order-preserving poly-line drawing 
of width $R(T)$ where the root is at the top left or top right corner.

Moreover, we can create such a drawing with at most 1 bend per edge.
Alternatively, we can create such a drawing with at most 3 bends
per edge and height at most $2n-1$.
\end{lemma}
\begin{proof}
We proceed by induction on the (graph-theoretic) height of $T$. 
The claim clearly holds if
$T$ is a single node since $R(T)=1$ and $T$ can be drawn with width 1
and height $1=2n-1$.    For the step, let
$c_1,\dots,c_d$ be the children of the root $u_r$ 
from left to right.  Recursively find a drawing $\Gamma_{c_i}$ of $T_{c_i}$ 
with width $R(T_{c_i})$.

Since $R(T)=W$, it has a corner-$W$-witness by
Lemma~\ref{lem:rank2corner}.
We assume that this is a left-corner-$W$-witness;
the construction is symmetric (and yields a drawing with the root
at the top right corner) if there is a right-corner-$W$-witness.  
So we have a sequence
$\sigma(W')<\dots<\sigma(W)$
(for some $1\leq W'\leq W+1$)
such that (C1) and (C2) hold.  Declare a child to be {\em big}
if its index is $\sigma(w)$ for some $W'\leq w\leq W$ and {\em small}
otherwise.  

Place the root at the top left corner.  
We place the children in two steps: first place the
small children (and start poly-lines for the edges to big children), and
then place the big children.  See the figure below for an example.

\medskip{\em Phase (1): } 
We parse the children in order $c_d,c_{d-1},\dots$.
Presume that $c_d,\dots,c_{j+1}$ have already been handled for some
$2\leq j\leq d$, and $Y$ is the 
lowest $y$-coordinate that has been used for them.   
Place a bend for
$(u_r,c_j)$ in column $2$ with $y$-coordinate $Y-1$.%
\footnotemark
\footnotetext{This bend can often be omitted, e.g.~if $c_j$ is small and at 
the top left
corner of $\Gamma_{c_j}$, but we show them in the figure for consistency.} 
All edges
$(u_r,c_k)$ with $k>j$ received bends in column $2$ at
larger $y$-coordinate, so this respects the order of edges around $u_r$.

Assume first that $c_j$ is a small child, say $\sigma(w-1)<j<\sigma(w)$
for some $W'\leq  w \leq W+1$.  Place $\Gamma_{c_j}$ in rows
$Y{-}2$ and below, and within columns $2,\dots,w-1$.  This fits since
by (C2) the rank of $c_j$ is at most $w-2$, and so $\Gamma_{c_j}$
occupies at most $w-2$ columns.
We can connect $c_j$ to the bend for edge $(u_r,c_j)$ with a straight-line
segment since $c_j$ is in the top row of $\Gamma_{c_j}$, and hence one row
below the bend.

\noindent
\begin{minipage}{0.29\linewidth}
\includegraphics[width=1.0\linewidth, page=1]{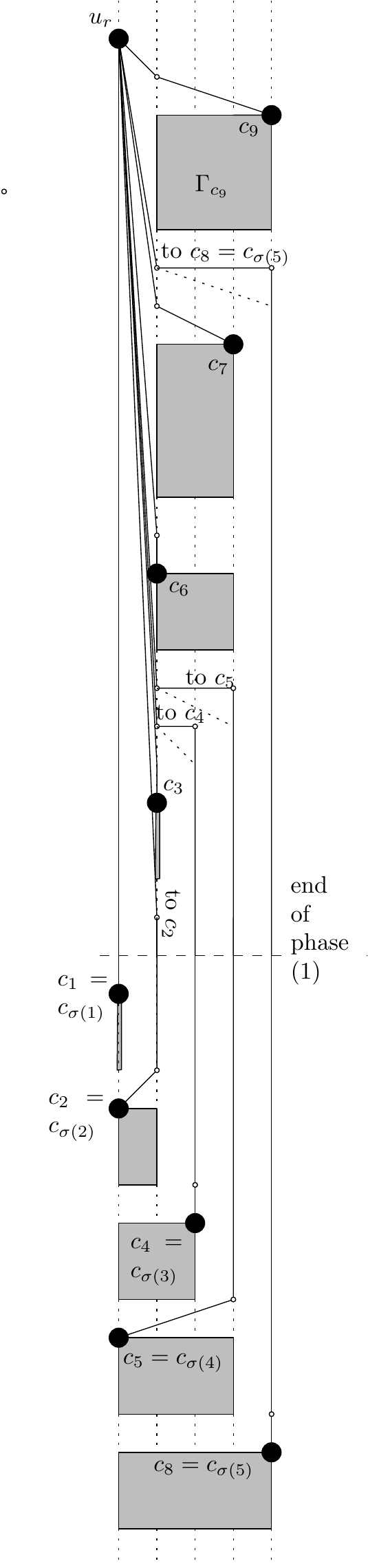}
\end{minipage}
\begin{minipage}{0.7\linewidth}
Now assume that $c_j$ is a big child, say $j=\sigma(w)$ for some
$W'\leq w\leq W$.  Place another bend for edge $(u_r,c_j)$ at point 
$(w,Y-1)$ and connect it horizontally to the bend at $(2,Y-1)$.  
Reserve the downward ray from this bend in column $w$ for this edge;
by construction no small child placed later will intersect
this ray.

This continues until we are left with $c_1$.  
Assign the downward ray in column 1 from the root to $c_1$,
and if $c_1$ is small, then place $\Gamma_{c_1}$ in
columns $1,\dots,W'-1$.

We have created some horizontal edges, and so the drawing,
while upward, is not strictly-upward.  We can make it strictly-upward
by re-locating the second bend for each edge to a big child to one 
row below, i.e., within the ray reserved for that edge.

\medskip{\em Phase (2): }
At this point all drawings of small children are placed, and the edge
to each big child $c_{\sigma(w)}$ is routed up to a 
vertical downward ray in column $w$. 
Place $\Gamma_{c_{\sigma(W')}},\Gamma_{c_{\sigma(W'+1)}},\dots,
\Gamma_{c_{\sigma(W)}}$, in this order from top to bottom,
below the drawing and flush left with column 1.
For $w\in \{W',\dots,W-1\}$,
since $c_{\sigma(w)}$ has rank $w$, its drawing has width
$w$ and will
not intersect the rays to $c_{\sigma(w+1)},\dots,c_{\sigma(W)}$.
By inserting a bend (if needed) in the row just above $c_{\sigma(w)}$,
we can complete the drawing of $(u_r,c_{\sigma(w)})$.

\medskip{\em Height-bound: }
Observe that every row of the drawing
contains the root, or
intersects some drawing $\Gamma_{c_i}$, or contains the first bend
of the edge $(u_r,c_i)$ for some child $c_i$.  Hence the total height
is at most
$1+\sum_{i=1}^d \mbox{(height of $\Gamma_{c_i}$)} + d$,
which by induction is at most
$1 + \sum_{i=1}^d (2n(T_{c_i}) - 1) + d
= 2n-1$. 

\medskip{\em Reducing bends: }  Every edge from $u_r$ to
a small child is drawn with one bend.  For a big child $c_{\sigma(w)}$,
the edge from $u_r$ may have up to three bends.  However, its poly-line
consists of at most two $x$-monotone parts:  from $u_r$ to column $w$,
and from column $w$ to $c_{\sigma(w)}$.  After subdividing at a point
in column $w$, we hence obtain a tree drawing where all edges are
$x$-monotone.  It is known \cite{EFL96,PT04} that such a drawing can be
turned into a straight-line drawing without increasing the width. 
Neither of these references discusses whether strictly upward drawings
remain strictly upward, but it is not hard to see that this can be
done, essentially by ``moving subtrees down'' sufficiently far.
We hence obtain a drawing with one bend per edge, at the cost
of increasing the height.
\end{minipage}
\end{proof}

\section{From drawing to rank-witness}
\label{sec:draw2rank}

\begin{lemma}
\label{lem:draw2rank}
If $T$ has an upward order-preserving poly-line drawing $\Gamma$ of width $W$, 
then $R(T)\leq W$.
Moreover, if $T$ is not a single node, 
then $T$ has a rank-$W$-witness for which coordinate $X$
equals the $x$-coordinate of the root.
\end{lemma}
\begin{proof}
If $T$ is a single node then $R(T)=1\leq W$ and the claim holds.
So assume that the root $u_r$
has children $c_1,\dots,c_d$ for some $d\geq 1$, and let $X$ be
the $x$-coordinate of $u_r$.
If there exists no edge that leaves $u_r$ vertically, then modify
$\Gamma$ slightly as follows.   Let $c_i$ be the last child (in the
order of children) for which the edge $(u_r,c_i)$ leaves $u_r$ to
the left of the vertical ray downwards from $u_r$.
(If there is no such child, then instead take the first child leaving
right of the ray.)  Re-route the edge $(u_r,c_i)$ so that it goes vertically
downward from $u_r$ for a brief while, then has a bend, and then connects
to where the old route crosses column $X{-}1$ (respectively $X{+}1$) for
the first time.  This adds no crossing and no width.
So we may assume that one edge leaves $u_r$ vertically; 
set $c_v$ to be the corresponding child. 

To classify each child $c$ as big or small, we study the induced drawing of its
subtree.  Let $\Gamma_c$ be the drawing of $T_c$ induced by $\Gamma$.
Let $\Gamma_c^+$ be $\Gamma_c$ together with the poly-line representing 
edge $(u_r,c)$, but excluding the point of $u_r$.
We declare $c$ to be big if $\Gamma_c^+$ contains a point in column $X$
and small otherwise.  With this $c_v$ is always a big child as desired.
The goal is to show that this classification as big/small, coordinate $X$,
and index $v$ satisfies the conditions for a rank-$W$-witness.

\myparagraph{Condition (R1$\ell$) and (R1r):}  
We only prove (R1$\ell$) here; (R1r) is similar.
So we must show that at most $X-1$ big children are left of $c_v$.
Consider Fig.~\ref{fig:bypass}(left).  
Let $q$ be any point below $u_r$ on the vertical segment of edge $(u_r,c_v)$.
Let $c_i$ be any big child strictly left of $c_v$.  Since 
the drawing is order-preserving, edge $(u_r,c_i)$ start towards $x$-coordinates
less than $X$.  Since $c_i$ is big, drawing $\Gamma_{c_i}^+$
contains a point with $x$-coordinate $X$; let $p_i$ be the topmost such point.
Due to the vertical line-segment $\overline{u_r q}$,
point $p_i$ is below $q$.
Let $P_i$ be the poly-line
within $\Gamma_{c_i}^+$ that connects $u_r$ to $p_i$; this exists
since $\Gamma_{c_i}^+$ is a drawing of a connected subtree.
All points in $P_i$ have $x$-coordinate at most $X$
by choice of $p_i$ and since the drawing is upward. 

If there are $k$ big children strictly left of $c_v$ then we hence obtain $k$
poly-lines $P_1,\dots,P_k$, which are disjoint except at $u_r$
and reside within columns $1,\dots,X$.  They all bypass point $q$ in
the sense that they begin above $q$ (in the same column) and end below $q$
(in the same column).  One can argue
(details are in Section~\ref{sec:bypass}) that each poly-line  requires a 
column
distinct from the one containing $q$ or used for the other poly-lines.
Since point $q$ and the poly-lines are all within columns $1,\dots,X$, 
this shows $k\leq X-1$ as desired.

\begin{figure}[ht]
\hspace*{\fill}
\includegraphics[width=35mm,page=3]{bypass.pdf}
\hspace*{\fill}
\includegraphics[width=35mm,page=2]{bypass.pdf}
\hspace*{\fill}
\includegraphics[width=35mm,page=1]{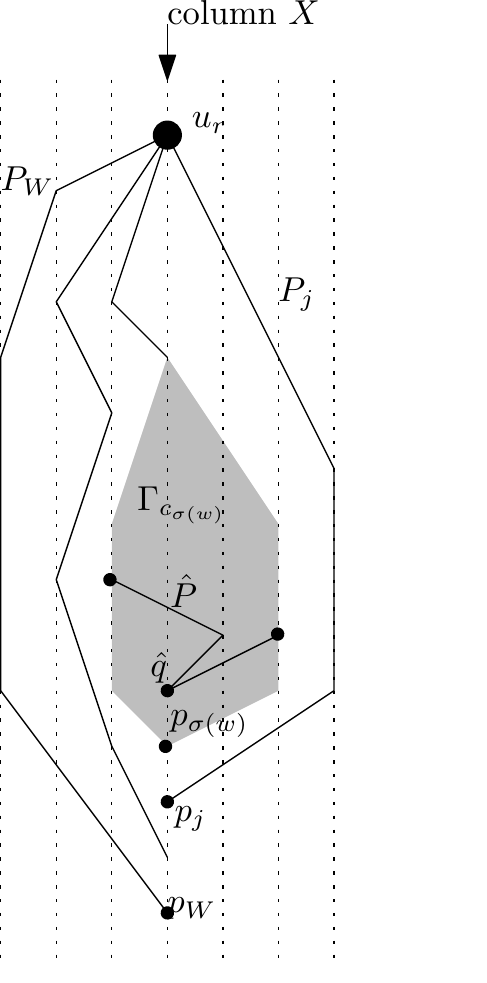}
\hspace*{\fill}
\caption{Bypassing lines.}
\label{fig:bypass}
\end{figure}

\myparagraph{Conditions (R2$\ell$) and (R2r):}  
We only prove (R2$\ell$) here; (R2r) is similar.
So we must show that any small child $c_i$ left of $c_v$ has rank
at most $X-1-\ell_i$.    We do this by finding a poly-line
for each big child left of $c_i$ that bypasses $\Gamma_{c_i}$ in some sense.  
These poly-lines block $\ell_i$ columns, leaving $X-1-\ell_i$
columns for $\Gamma_{c_i}$, hence $R(T_{c_i})\leq X-1-\ell_i$
by induction.

Consider Fig.~\ref{fig:bypass}(middle).  
Let $p_i$ be the leftmost point of drawing $\Gamma_{c_i}^+$, breaking
ties arbitrarily.  Let $q_i$ be the point where the initial line
segment of $(u_r,c_i)$ intersects column $X-1$; this must exist
since edge $(u_r,c_v)$ leaves $u_r$ vertically and $(u_r,c_i)$ must
leave $u_r$ to the left of this.
Let $P_i$ be the poly-line from $q_i$ to $p_i$ within
drawing $\Gamma_{c_i}^+$.  Since $c_i$ is small, $P_i$ does not
use column $X$. 

Let $c_h$ be a big child to the left of $c_i$ and let $q_h$
be the point where the initial line segment of $(u_r,c_h)$ intersects
column $X-1$.  Since the drawing is order-preserving, $q_h$ is above $q_i$.
Since $c_h$ is big, drawing $\Gamma_{c_h}^+$
intersects column $X$, and in particular therefore has a line
segment $\overline{p_hp_h'}$ with $p_h$ in column $X-1$ and $p_h'$ in 
column $X$.  Since $\overline{p_hp_h'}$ must not intersect
$\overline{u_rq_i}$, $p_h$ must be below $q_i$.  Re-define $p_h$,
if necessary, to be the topmost point below $q_i$ where
$\Gamma_{c_h}^+$ intersects column $X-1$. Let $P_h$ be the
poly-line from $q_h$ to $p_h$ within $\Gamma_{c_h}^+$.  By choice of
$p_h$ and line segment $\overline{u_rq_i}$, poly-line $P_h$
is within coordinates $1,\dots,X-1$.

Repeating this for all $\ell_i$ big children left of $c_i$ gives $\ell_i$
poly-lines that reside within $1,\dots,X-1$ and that bypass $P_i$ in the
sense that they begin and end in column $X-1$, with one end above $q_i$ and 
the other below $q_i$.
Again one can show that these $\ell_i$ poly-lines
each require one column in $\{1,\dots,X-1\}$ that does not intersect $P_i$.
Therefore $P_i$ (and with it $\Gamma_{c_i}$) has width at most
$X-1-\ell_i$, so $R(T_{c_i})\leq X-1-\ell_i$ by induction.

\myparagraph{Condition (R3):}  To verify this condition, we 
extract rank-bounds from drawing $\Gamma$ as follows.
Let $p_W$ be the lowest point in column $X$ that is occupied by some element
of $\Gamma$.  
Due to the vertical segment of edge $(u_r,c_v)$, point $p_W$ 
is not the locus of the root.  Let $c_j$ be the child 
such that $\Gamma_{c_j}^+$ contains $p_W$; by definition $c_j$ is big.
Set $\sigma(W):=j$ and $\pi(c_j):=W$.  

Now presume we have found $\sigma(W),\sigma(W-1),\dots,\sigma(w+1)$ already
for some $w< W$.  Let $p_{w}$ be the lowest point in column $X$ that
is occupied by some element in $\Gamma$ but that does not belong to any
of $\Gamma^+_{c_{\sigma(W)}},\dots,\Gamma^+_{c_{\sigma(w+1)}}$.  If this point is
at $u_r$, then stop:  we have assigned
a rank-bound to all big children.
Else, let $c_j$ be the child such that $\Gamma_{c_j}^+$
contains $p_{w}$, set $\sigma(w):=j$ and $\pi(c_j):=w$, and repeat.

We must show that the chosen values are indeed rank-bounds,
i.e., $R(T_{c_{\sigma(w)}})\leq w$, for all $w$ where $\sigma(w)$ is defined. 
By induction it suffices to show that the width of 
$\Gamma_{c_{\sigma(w)}}$ is at most $w$.
Consider Fig.~\ref{fig:bypass}(right).  
Let $\hat{P}$ be the poly-line within $\Gamma_{c_{\sigma(w)}}$ that connects
a leftmost and rightmost point of $\Gamma_{c_{\sigma(w)}}$.  
Recall that with the rank-bounds we also found points
$p_W,p_{W-1},\dots,p_{w}$, where for $j>w$ point $p_j$ belongs to $\Gamma_{c_{\sigma(j)}}$,
has $x$-coordinate $X$ and is below $p_{j-1}$.  
For any $j>w$, let $P_j$ be the poly-line that connects
$u_r$ with point $p_j$ within $\Gamma_{c_{\sigma(j)}}^+$.
Poly-line $\hat{P}$ spans the width of $\Gamma_{c_{\sigma(w)}}$ and hence
must cross column $X$, say at point $\hat{q}$.  This crossing point cannot be
below $p_w$ due to choice of $p_w$ as the lowest point in column $X$
that is not in $\Gamma^+_{c_{\sigma(w+1)}},\dots,\Gamma^+_{c_{\sigma(W)}}$.
For any $j>w$ point $p_j$ is below
$p_w$ and hence also below $\hat{q}$.  On the other hand $\hat{P}$
does not contain $u_r$ (since it resides within $\Gamma_{c_{\sigma(w)}}$,
not $\Gamma_{c_{\sigma(w)}}^+$), and so $\hat{q}$ is below $u_r$.

We now have found $W-w$ poly-lines $P_{w+1},\dots,P_{W}$ that bypass $\hat{P}$
in the sense that $P_j$ connects $u_r$ (a point above $\hat{q}$) with $p_j$ (a point
below $\hat{q}$), and these poly-lines are node-disjoint from
$\hat{P}$ and from each other except at $u_r$.
Again one can show that each poly-line requires a column
of its own that does not contain $\hat{P}$.
Since there are
$W-w$ such poly-lines, and the drawing of $T$ has width $W$, therefore
$\hat{P}$ (and with it
$\Gamma_{c_{\sigma(w)}}$) has width at most $w$. 

\medskip
This proves that this classification, coordinate, and index
give a rank-$W$-witness, so $R(T)\leq W$ as desired.
\end{proof}

\subsection{Bypassing poly-lines}
\label{sec:bypassing}
\label{sec:bypass}

In the proof of Lemma~\ref{lem:draw2rank}, 
we repeatedly used that some set
of poly-lines bypasses another poly-line, and therefore each of
them requires a column of its own.  This is quite intuitive: 
many lower-bound arguments for planar graph drawing use 
arguments where so-called ``nested cycles'' each require two
additional columns (see e.g.~\cite{FPP88}).  However, the argument
is non-trivial for poly-lines since they are open-ended
curves and hence do not separate the drawing of the rest from the
``outside'', except under the special conditions that we called
bypassing.  The rest of this subsection gives the precise definition
and argument. 

We previously described three different situations for bypassing,
but one easily checks that the following definition
encompasses them all: 

\begin{definition}
Let $\hat{P},P_1,\dots,P_k$ be a set of poly-lines that are
disjoint except that ends of $P_1,\dots,P_k$ may coincide. We say
that $P_1,\dots,P_k$ {\em bypass} $\hat{P}$ if there exists
a point $\hat{q}$ in $\hat{P}$ such that for all $i=1,\dots,k$
poly-line $P_i$ begins at a point above $\hat{q}$
and ends at a point below $\hat{q}$.

Here, a {\em point above[below] $\hat{q}$} means a point with the
same $x$-coordinate as $\hat{q}$ 
and with $y$-coordinate strictly larger[smaller]
than the one of $\hat{q}$.
\end{definition}

Recall that for poly-lines the endpoints and all bends  must have
integral $x$-coordinates, and that we measure the width of a set
of poly-lines by the minimum number of consecutive columns that contain them.
Let $x_{\min}(P)$ and $x_{\max}(P)$ be the minimum and maximum $x$-coordinate
of points in poly-line $P$.

\begin{lemma}
Let $P_1,\dots,P_k$ be a set of poly-lines that bypass a poly-line $\hat{P}$.
If these poly-lines all reside within columns $1,\dots,W$, then
$$ W \geq \left(x_{\max}(\hat{P})-x_{\min}(\hat{P})+1\right) + k $$
In other words, every bypassing poly-line requires one additional
column beyond the width occupied by $\hat{P}$.
\end{lemma}
\begin{proof}
We proceed by induction on $W$, with an inner induction on
the total number of bends in poly-lines $P_1,\dots,P_k$. 
Clearly $W\geq x_{\max}(\hat{P})-x_{\min}(\hat{P})+1$
since $\hat{P}$ alone occupies this many columns.
In the base case,   $W= x_{\max}(\hat{P})-x_{\min}(\hat{P})+1$,
which means that poly-line $\hat{P}$ extends from leftmost
to rightmost column.  Therefore $\hat{P}$ separates all points
above $\hat{q}$ from points below $\hat{q}$. This implies that
no poly-line $P_1$ exists since $P_1$ is
disjoint from $\hat{P}$ and hence cannot cross it.    Thus,
$k=0$ and the claim holds.

For the induction step $W> x_{\max}(\hat{P})-x_{\min}(\hat{P})+1$, 
so $\hat{P}$ does not span all columns. Say
$x_{\max}(\hat{P})<W$, so $\hat{P}$ is within columns $1,\dots,W-1$.  We have cases.

In the first case, at most one of $P_1,\dots,P_k$ intersects column $W$.
Say this poly-line (if one exists) is $P_k$.
Then $P_1,\dots,P_{k-1}$ all reside within
columns $1,\dots,W-1$, as does $\hat{P}$.  By induction therefore
$W-1 \geq  x_{\max}(\hat{P})-x_{\min}(\hat{P})+1+(k-1)$,
which proves the claim.

In the second case,
some poly-line $P_i$ contains three or more points in the column $X$ that
contains $\hat{q}$.
Then some strict sub-poly-line of $P_i$ connects a point in column $X$
above $\hat{q}$ with a point in column $X$ below $\hat{q}$. We can shorten $P_i$ to this
smaller poly-line without affect the conditions on bypassing. This
removes at least one bend from $P_i$ and 
the claim holds by induction.

\begin{figure}[ht]
\hspace*{\fill}
\includegraphics[width=40mm,page=2]{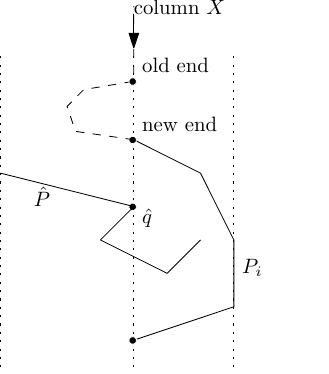}
\hspace*{\fill}
\includegraphics[width=40mm,page=3]{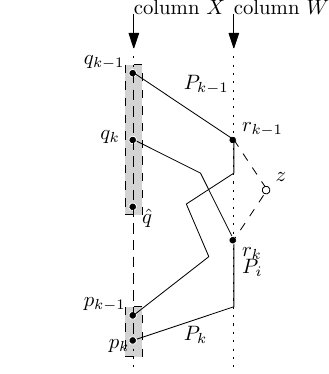}
\hspace*{\fill}
\caption{Bypassing poly-lines require extra columns.  
(Left) Pruning a path
that intersects column $X$ three times.  (Right) Finding a $K_4$-minor
if none of the previous cases applies.}
\label{fig:bypass_proof}
\end{figure}

Finally we argue that one of the above cases must apply.  Assume
for contradiction that two poly-lines, say $P_{k-1}$ and $P_k$,
both contain a point in column $W$.  Observe that $X<W$,
since column $X$ must intersect $\hat{P}$ due to point $\hat{q}$,
but $x_{\max}(\hat{P})<W$.  Since the second case does not
apply, each $P_i$ (for $i=k-1,k$) stays strictly right of $X$ 
except at its endpoints.  Hence $P_i$ starts at point $q_i$ in column $X$
above $\hat{q}$, connects to a point $r_i$ in column $W$, and then 
returns to point $p_i$ below $\hat{q}$ in column $X$, all the while
staying within $X+1,\dots,W$ except at the ends.
One can observe that this is impossible without a crossing.
Formally one proves this by creating
an outer-planar drawing of a $K_4$-minor as follows:
Consider the drawing induced by $P_k$ and $P_{k-1}$.
Connect the points in column $X$ with vertical edges in order, and
add a new node $z$ in column $W+1$ adjacent to $r_k$ and $r_{k-1}$.
See also Fig.~\ref{fig:bypass_proof}.  
This clearly maintains planarity
and all of $\{q_{k-1},q_k,\hat{q},p_{k-1},p_k,r_{k-1},r_k,z\}$ are on the outer-face.  Since 
$q_k$ and $q_{k-1}$ are strictly above $\hat{q}$ while $p_k$ and $p_{k-1}$
are strictly below, not all points with $x$-coordinate $X$ can coincide.
Since $P_{k-1}$ and $P_k$ are disjoint (except perhaps at their ends),
points $r_k$ and $r_{k-1}$ cannot coincide.  So this indeed gives an
outer-planar drawing of a minor of $K_4$, which is impossible.  So
one of the above cases must apply, and the claim
holds by induction.
\end{proof}

\section{Transforming rank-witnesses}
\label{sec:rank2corner}

The goal of this section is to prove Lemma~\ref{lem:rank2corner}, i.e.,  
to find a corner-$W$-witness for a tree of rank $W$.
We go further and show a
chain of equivalences, which also gives rise to a fast algorithm
to test the existence of a corner-$W$-witness.

\begin{lemma}
\label{lem:equivalent}
Let $T$ be a tree for which the root has $d\geq 1$ children, and let $W\geq 1$ be an integer.  
The following are equivalent:
\begin{enumerate}
\item $T$ has a rank-$W$-witness.
\item $T$ has a rank-$W$-witness with $X\in \{1,W\}$.
\item $T$ has a rank-$W$-witness with $v\in \{1,d\}$
\item Algorithm {\sc TestLeft}$(W)$ (given below) returns with success 
	or algorithm {\sc TestRight}$(W)$ returns with success.
\item $T$ has a left-corner-$W$-witness or a right-corner-$W$-witness.
\item $T$ has a corner-$W$-witness.
\end{enumerate}
\end{lemma}
\begin{proof}
We give the easy implications first and then prove the harder ones in
separate lemmas.  
\begin{itemize}
\item (1)$\Rightarrow$(2) will be proved in Lemma~\ref{lem:X1orW}.
\item (2)$\Rightarrow$(3) holds automatically for the same rank-$W$-witness.
	Say we have a rank-$W$-witness with $X=1$ (the case $X=W$ is similar).
	If $v>1$ then by (R1$\ell$) no big children are left of $c_v$, so $c_1$ 	
	must be a small child.  But then by (R2$\ell$) child $c_1$ must have  
	rank at most $X-1=0$, an impossibility.  So $v=1$. 
\item (3)$\Rightarrow$(4) will be proved in Lemma~\ref{lem:testFailure}.
\item (4)$\Rightarrow$(5) will be proved in Lemma~\ref{lem:testSuccess}.
\item (5)$\Rightarrow$(6) holds by definition of corner-$W$-witness.
\item (6)$\Rightarrow$(2) could be proven directly, but a simpler indirect proof is
	that Lemma~\ref{lem:corner2draw} shows how to extract a drawing of width $W$
	from the corner-$W$-witness, and Lemma~\ref{lem:draw2rank} shows how to 
	extract a rank-$W$-witness from this drawing.  In the drawing, the root
	is at the top left or top right corner, and hence in the rank-$W$-witness
	we have $X=1$ or $X=W$.
\item (2)$\Rightarrow$(1) holds trivially.
\end{itemize}
\end{proof}

\begin{algorithm}[ht]
\begin{tabbing}
\hspace*{2cm}\=\hspace*{2cm}\=\kill
// $T$ is a tree with children $c_1,\dots,c_d$, $d\geq 1$, $W\geq 1$ \\
Let $i$ be the maximal index such that $R(T_{c_i})\geq W$ \\
\IF{no such $i$ exists}  \RETURN ``success'' \\
\IF{$R(T_{c_i})>W$} \RETURN ``failure'' \\
Now $c_i$ is the rightmost child with $R(T_{c_i})=W$. \\
Initialize $\sigma(W)$ to be $i$, $w$ to be $W$ and decrease $i$ \\
\LOOP \\
\>\WHILE{$i>0$ \AND $R(T_{c_i})\leq w-2$} decrease $i$ \\
\>\IF{$i==0$} set $W':=w$ and \RETURN ``success'' \\
\>\IF{$R(T_{c_i})\geq w$} set $W':=w$ and \RETURN ``failure'' \\
\>Now $c_i$ is a child with $R(T_{c_i})=w-1$ and $i<\sigma(w)<\dots<\sigma(W)$\\
\>Set $\sigma(w-1)$ to be $i$ and decrease both $w$ and $i$. \\
\ENDLOOP 
\end{tabbing}
\vspace*{-3mm}
\caption{{\sc TestLeft}$(T,W)$}
\label{algo:test}
\end{algorithm}

Algorithm~\ref{algo:test} gives the algorithm {\sc TestLeft} that tests
whether a tree $T$ has a left-corner-$W$-witness.  
We give now the lemmas that show its correctness.
The corresponding results for algorithm  
{\sc TestRight} for right-corner-$W$-witnesses are in the appendix.	

\begin{lemma}
\label{lem:testSuccess}
Assume algorithm {\sc TestLeft} returns with ``success''.  Then $T$ has
a left-corner-$W$-witness.
\end{lemma}
\begin{proof}
There are two possible situations in which {\sc TestLeft} returns success.
One possibility is that no child has 
rank $W$ or higher; then by Observation~\ref{obs:W+1} we have a 
left-corner-$W$-witness. 
The other possibility is that the algorithm reached $i=0$ and therefore
found a value $W'$ and indices $\sigma(W')<\sigma(W'+1)<\dots<\sigma(W)$
with $R(T_{c_{\sigma(w)}})=w$ for all $W'\leq w\leq W$.    Let $c_i$ be a child 
that was skipped when assigning $\sigma(.)$, i.e., $\sigma(w-1)<i<\sigma(w)$
for some $W'\leq w\leq W$  (where as before $\sigma(W'-1):=0)$.  We skipped
this child because  has rank at most $w-2$, so (C2) holds for $c_i$.  Also,
all children to the right of $c_{\sigma(W)}$ have rank at most $W-1$, so
again (C2) holds.  So we found a left-corner-$W$-witness.
\end{proof}

\begin{lemma}
\label{lem:testFailure}
Assume algorithm {\sc TestLeft} returns with ``failure''.  Then $T$ has
no rank-$W$-witness with $v=1$.
\end{lemma}
\begin{proof}
There are two possible situations in which {\sc TestLeft} returns failure.
One possibility is that some child has rank $W+1$ or higher; then by Observation~\ref{obs:W} 
no rank-$W$-witness can exist.  The 
other possibility is that the algorithm reached some $i>0$ with $R(T_{c_i})\geq W'$
and indices $\sigma(W')<\sigma(W'+1)<\dots<\sigma(W)$ where $c_{\sigma(w)}$ has rank $w$ for all $W'\leq w\leq W$.    
Assume for contradiction that a rank-$W$-witness with $v=1$ exists.
We claim that children $c_i,c_{\sigma(W')},\dots,c_{\sigma(W)}$ must all be big.
This is obvious for $c_{\sigma(W)}$: By $v=1$ this child is right of the
vertical child, and by (R2r) it cannot be small since its rank is $W$.
Now $c_{\sigma(W-1)}$ has at least one big child to its right, and it is
also to the right of the vertical child, so since its rank is $W-1$ and using (R2r) shows that it, too, must be big.
Repeating the argument show that children $c_i,c_{\sigma(W')},\dots,c_{\sigma(W)}$ are all big.
But this gives $W-W'+2$ big children with ranks in $\{W',\dots,W\}$, which means that
it is impossible to assign rank-bounds and satisfy (R3).  Hence no rank-$W$-witness with $v=1$ can exist.
\end{proof}

The final step is hence to show that the coordinate of a rank-$W$-witness
can be ``pushed into a corner''.

\begin{lemma}
\label{lem:X1orW}
Let $T$ be a tree.
If $W:=R(T)\geq 2$, then $T$ has a rank-$W$-witness with
$X=1$ or $X=W$.
\end{lemma}
\begin{proof}
If all children have rank at most $W-1$, then such a witness
is easily constructed by setting $X=v=1$ and declaring all children except
$c_1$ to be small.  We leave it to the reader to verify the conditions.

So assume some child $c_m$ has rank
$W$. Fix any rank-$W$-witness of $T$, and assume $1<X<W$ for
its coordinate, otherwise we are done.
By (R2$\ell$) and (R2r), any small child has rank at most
$\max\{X-1,W-X\}\leq W-2$ since $1<X<W$.  So any child
of rank $W-1$ or $W$ is big, and
by (R3) we can have at most one child $c_s$ with rank $W-1$.

Assume that $c_s$ either does not exist or is strictly right of $c_m$.
Create a rank-$W$-witness using $X=1$ and $v=1$ and
declaring $c_1$ and $c_m$ to be big and all other children
to be small.  Verify the conditions
for this new witness as follows.  (R3) holds since we have
at most two big children, and only one of them has rank $W$.
(R1$\ell$) and (R2$\ell$)
hold trivially since $v=1$.  (R1r) holds since at most
$1\leq W-1$ big children are right of $c_1$.  (R2r)
holds for $i>m$ since then $r_i=0$ and $c_i$ has rank
at most $W-1$.  It also holds for $1<i<m$ since then $r_i=1$
and $c_i$ has rank at most $W-2$ since $c_s$ (if it exists)
is strictly right of $c_m$.

This creates a rank-$W$-witness with $X=1$ if $c_s$ does not
exist or is strictly right of $c_m$.  If $c_s$ is strictly
left of $c_m$, then
similarly create a rank-$W$-witness with $X=W$ and $v=d$.
\end{proof}

So not only can any rank-$W$-witness be turned into
a corner-$W$-witness (which proves Lemma~\ref{lem:rank2corner}),
but with the proof we also get an algorithm to test whether
such a witness exists.

\begin{lemma}
For any tree $T$, $R(T)$ can be computed in linear time.
In the same time we can also find a corner-witness (for the 
respective rank) for each rooted subtree of $T$.
\end{lemma}
\begin{proof}
If $T$ has one node, then $R(T)=1$ and we are done.  So
assume $n\geq 2$ and we have already recursively computed
ranks and corner-witnesses for the children.  Let $W$ be the
maximal rank among the children.  Run {\sc TestLeft($T,W$)}
and {\sc TestRight($T,W$)} to test whether $T$ has a 
corner-$W$-witness.  If one of them succeeds,
then $R(T)=W$ and we have found the corner-witness.  Otherwise
$R(T)\geq W+1$ by Lemma~\ref{lem:equivalent}, and we know $R(T)\leq W+1$ and can find
the left-corner-$(W+1)$-witness using Observation~\ref{obs:W+1}.
This computation takes $O(\deg(v))$ time for each node $v$,
and hence $O(n)$ time total.
\end{proof}

With this, all ingredients for Theorem~\ref{thm:main} have been
assembled and the theorem holds.  We also note that our proof
shows that for order-preserving poly-line drawings, it makes
no difference for the width whether we demand upward or strictly-upward
drawings.  The extraction of the rank-$W$-witness from a drawing
(Lemma~\ref{lem:draw2rank}) works even if the drawing has horizontal
edges, while the construction of the drawing (Lemma~\ref{lem:corner2draw})
creates strictly-upward drawings.

\section{Straight-line drawings?}
\label{sec:straight}

We showed that the rank exactly describes the optimum width of
{\em poly-line} upward order-preserving drawings.  A natural question
is whether this also describes the optimum width of {\em ideal} drawings
where additionally we require edges to be straight-line. 
The answer is ``no''.

\begin{theorem}
\label{thm:ordered_straightline_quaternary}
\label{thm:quaternary}
The tree in Fig.~\ref{fig:quintinary}(a)
has a planar strictly-upward 
order-preserving poly-line drawing of width $2$, but no 
ideal drawing of width $2$.
\end{theorem}

Nevertheless, might there be a similar algorithm to compute optimum-width
straight-line drawings?  This question remains open, but we can show that
one key ingredient will fail:  There do not always exist optimum-width
drawings where the root is at a corner.

\begin{theorem}
\label{thm:not_corner}
The tree in Fig.~\ref{fig:quintinary}(b)
has a planar upward order-preserving straight-line
drawing of width 3, but in any such drawing the root has to be in 
the middle column. 
\end{theorem}

The proofs of these theorems are in Appendix~\ref{ap:straight}.
The trees in these theorems are quaternary (i.e., all nodes have degree 4 
or less) and this is tight: any
ternary tree $T$ has a straight-line order-preserving drawings 
with the root in a corner and width $rpw(T)=R(T)$ \cite{OPTI}.

\begin{figure}[ht]
\begin{tabular}{ccccccc}
\includegraphics[width=12mm]{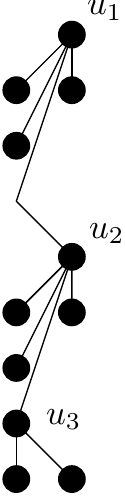} & \hspace*{5mm} &
\includegraphics[width=15mm,page=1]{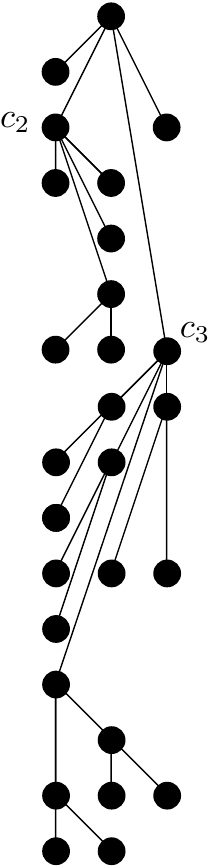} & \hspace*{5mm} &
\includegraphics[width=50mm,page=1,trim=0 60 0 0,clip]{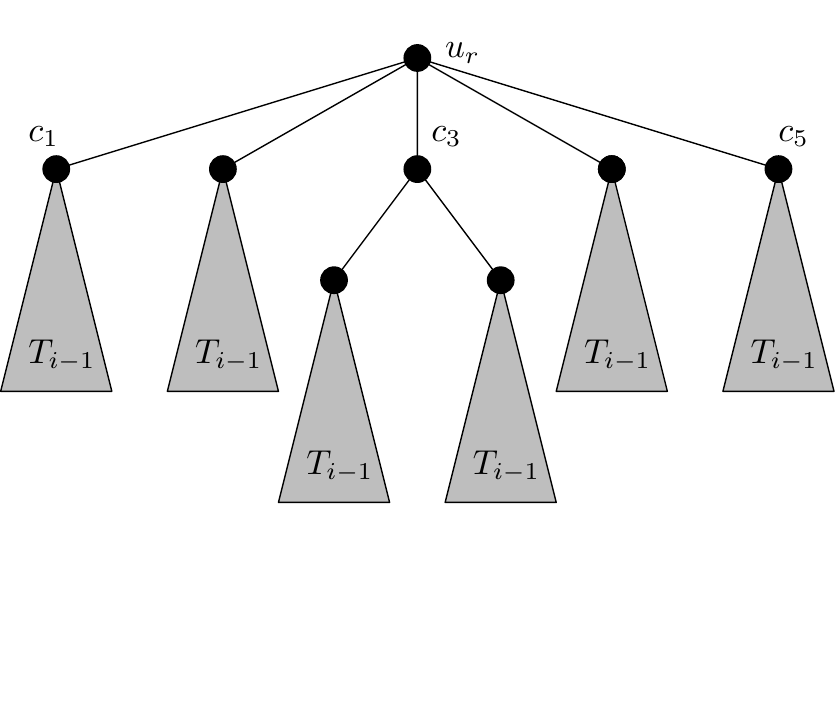} & \hspace*{5mm} &
\includegraphics[width=30mm,page=2,trim=20 0 90 0,clip]{quintinary.pdf} \\
(a) && (b) && (c) && (d) \\
\end{tabular}
\caption{(a) A tree that cannot be drawn straight-line with the same width.
(b) A tree that cannot be drawn straight-line with the root at the corner
and the same width.  (c) and (d): A tree where order-preserving drawings require
nearly twice as much width as unordered drawings.}
\label{fig:quaternary}
\label{fig:quintinary}
\label{fig:orderedNotStraight}
\end{figure}

\section{Comparing rooted pathwidth and rank}

It is not hard to see (details are in the appendix) that any tree
has rooted pathwidth at most $\log(n+1)$ and rank at most $\log n+1$.
Since these two numbers are very close, one might wonder whether 
rooted pathwidth and rank are always within a constant of each other?
This is not the case:  The tree in Figure~\ref{fig:quintinary}(c) and (d)
has rooted pathwidth $i$, but rank $2i-1$ (see the appendix for a proof),
and so it requires almost twice as much width in an order-preserving drawing
compared to an unordered one.  This tree has degree 5; one
can show (see \cite{OPTI}) that for trees with degree at most 4 the two 
parameters coincide.

\section{Conclusion}
\label{sec:conclusion}

In this paper, we gave two linear-time algorithms for tree drawings.  The first
finds a planar strictly-upward straight-line drawing, and the second finds
a planar strictly-upward poly-line drawing that respects the given order
of the children at all nodes.  Both algorithm achieve the optimal width among all such
drawings.    
Many open problems remain: 
\begin{itemize}
\item Can we compute {\em ideal drawings} of
	optimum width?  The examples of Section~\ref{sec:straight}
	suggest that this requires a different approach.
\item Can we find tree drawings 
	that have optimal {\em area}, or is this
	NP-hard?    
	(The question could be asked for many different types of drawings,
	such as order-preserving or not, or straight-line or not, upward
	or not.)
\item
	Can we at least prove the conjecture in \cite{DF14}
	that every tree has a strictly-upward
	straight-line order-preserving drawing of area $O(n\log n)$?
	The best currently known bound is $O(n4^{\sqrt{2\log n}})$ 
	\cite{Chan02} or $O(\Delta n\log n)$ for a tree
	with maximum degree $\Delta$ \cite{OPTI}.
\end{itemize}

\bibliographystyle{plain}
\bibliography{../../bib/journal,../../bib/full,../../bib/gd,../../bib/papers}

\begin{appendix}
\section{Rooted pathwidth and other parameters}
\label{sec:related}

In this section we study more properties of the rooted pathwidth,
and in particular, relate it to some other graph parameters that
have been used for tree drawings.

\subsection{Logarithmic bound:}

\begin{lemma}
\label{lem:hpd}
Any tree $T$ with $rpw(T)=r$ has at least $2^r-1$ nodes and at
least $2^{r-1}$ leaves.  In particular, $rpw(T)\leq \log(n+1)$.
\end{lemma}
\begin{proof}
Clearly this holds if $T$ is a single node and $r=1$, so assume the root
has children. If one child $c$ has $rpw(T_c)=r$, 
then the claim holds by induction for $T_c$ and hence also for $T$.  
Otherwise, by definition of $rpw(T)$ there must be at least two 
children $c_1,c_2$
with $rpw(T_{c_j})=r-1$ for $i=1,2$.  Applying induction to both and combining
the bounds (and adding the root) gives the result.
\end{proof}

This bound is tight for the complete binary tree with height $h$
(where a single-node tree is considered to have height $1$).
Such a tree has $n=2^h-1$ nodes and rooted pathwidth $h=\log(n+1)$.

\subsection{Root-to-leaf paths:}
Let $P$ be a root-to-leaf path in $T$, i.e., a path from the root to
some arbitrary leaf.  Removing $P$ splits $T$ into subtrees.  We now
claim that if we choose $P$ suitably, then all these subtrees have
smaller rooted pathwidth, and show:

\addtocounter{observation}{-1}
\begin{observation}
We have
$$
rpw(T) = \left\{
\begin{minipage}{70mm}
$
\begin{array}{ll}
1 & \mbox{if $T$ is a rooted path} \\
\min_{P} \max_{T' \subset T-P} \left\{ 1+ rpw(T')\right\} & \mbox{otherwise}
\end{array}
$
\end{minipage}
\right.
$$
\end{observation}
\begin{proof}
We show `$\geq$' by induction on the height of the tree.  Clearly the
claim holds for a single-node tree, so assume the root has children.
Let $P$ be the path
obtained by going from the root to the rpw-heaviest child,
and from there to its rpw-heaviest child, etc., until we reach a leaf.
Any subtree $T'$ of $T-P$ then corresponds to tree $T_c$ for a node $c$ 
which is not on $P$, but its parent $v$ is on $P$.  Since $c$ was not the
rpw-heaviest child of $v$, we have $rpw(T_c)< rpw(T_v)\leq rpw(T)$, hence
$\allowbreak \max_{T' \subset T-P} \left\{ 1+ rpw(T')\right\}\leq rpw(T)$.
The minimum over all choices of path can only be smaller. 

For the other direction, let $P$ be the path that minimizes
$r:=\allowbreak \max_{T' \subset T-P} \{ 1\allowbreak +$
$rpw(T')\}$,
and let $c_h$ be the child of the  root that belongs to $P$.  
Then any child $c\neq c_h$
of the root gives rise to a subtree $T'=T_c$ of $T-P$, hence 
$1+rpw(T_c)\leq r$.  Also, $rpw(T_{c_h})\leq r$ by induction,
since $P$ (minus the root) can be used as a path for $T_{c_h}$.
Therefore $\max_{c} \left\{ rpw(T_c)+\chi(c\neq c_h) \right\}\allowbreak\leq r$
and the minimum over all choices of $c_h$ can only be smaller.
\end{proof}

\subsection{Pathwidth:}
The {\em pathwidth} $pw(G)$ of a graph $G$ is a well-known graph parameter;
it is the smallest integer $k$ such that $G$ is a subgraph of a 
$(k+1)$-colorable interval graphs.  For trees, the pathwidth can also be 
described via a decomposition into paths; see \cite{EST94,Sud04}.  Namely
$$
pw(T) = \left\{
\begin{minipage}{70mm}
$
\begin{array}{ll}
0 & \mbox{if $T$ is a single node} \\
\min_{P} \max_{T' \subset T-P} \left\{ 1+pw(T')\right\} & \mbox{otherwise}
\end{array}
$
\end{minipage}
\right.
$$
where the minimum is taken over all paths $P$.  
As in \cite{Sud04} we
call the path $P$ where the minimum is achieved 
the {\em main path}.
Note that the recursive
formula is the same as in Observation~\ref{obs:rpw}, except
that the path $P$ is not restricted to end at the root.  
A simple proof by induction hence shows that $pw(T)\leq rpw(T)$.
At the other end, we can show:

\begin{lemma}
For any rooted tree $T$, we have $rpw(T) \leq 2pw(T)+1$.
\end{lemma}
\begin{proof}
%
This was essentially shown by Suderman
\cite{Sud04} (he also gives credit to Dujmovi\'{c} and Wood)
without using the term ``rooted pathwidth''.  In the
second half of the proof of his Lemma 7, he creates tree-drawings
of height at most $2pw(T)$.  An inspection of the construction
shows that it gives upward drawing after $90^\circ$ rotation,
except at subtrees with pathwidth 1 (which could be drawn upright if
we allowed one extra unit.)  By Lemma~\ref{lem:lower} hence
$rpw(T)\leq 2pw(T)+1$.

For completeness' sake, we give here an independent proof of this result,
using the same idea as implicit in Suderman's algorithm \cite{Sud04}.
If $pw(T)=0$, then $T$ is a single node and $rpw(T)=1=2pw(T)+1$,
so the claim holds.  If $pw(T)\geq 1$, then let $P$ be a main path of $T$.
See also Fig.~\ref{fig:tree_paths}.  We may, after expanding $P$ if
needed, assume that the ends of $P$ are at the root or at a leaf.
Let $v$ be the node of $P$ that is closest to the root, and write
$P=P_1-v-P_2$ for two paths $P_1$ and $P_2$.  By definition any
subtree $T'$ of $T-P$ has $pw(T')\leq pw(T)-1$ and therefore
$rpw(T')\leq 2pw(T)-1$.

Let $P_0$ be the path from the root to $v$.
Let $P':=P_0-v-P_1$ consists of
the path from the root to $v$, followed by one part of the main path of $T$.
We use $P'$ as the path in Observation~\ref{obs:rpw}, and hence must
study the rooted pathwidth of 
any subtree $T'$ of $T-P'$.  If $T'$ is also a subtree of $T-P$,
then as argued above $rpw(T')\leq 2pw(T)-1$.  If $T'$ is not a 
subtree of $T-P$, then $T'$ necessarily must contain $P_2$; call
this subtree $T_2$.

One can show that $rpw(T_2)\leq 2pw(T)$ as follows.  Use path $P_2$
as the path in Observation~\ref{obs:rpw}; we hence must study the
rooted pathwidth of any subtree $T''$ of $T_2-P_2$.  But any such subtree
contains no nodes of $P$ and hence is a subtree of $T-P$.  By 
the above discussion therefore $rpw(T'')\leq 2pw(T)-1$.  Therefore
$rpw(T_2)\leq \max_{T''} \left\{ 1+rpw(T'') \right\} \leq 2pw(T)$.

Putting it all together, we know that $rpw(T')\leq 2pw(T)$ for all
subtrees $T'$ of $T-P$, and by Observation~\ref{obs:rpw} therefore
$rpw(T)\leq 2pw(T)+1$.
\end{proof}

\begin{figure}[ht]
\hspace{\fill}
\includegraphics[height=50mm]{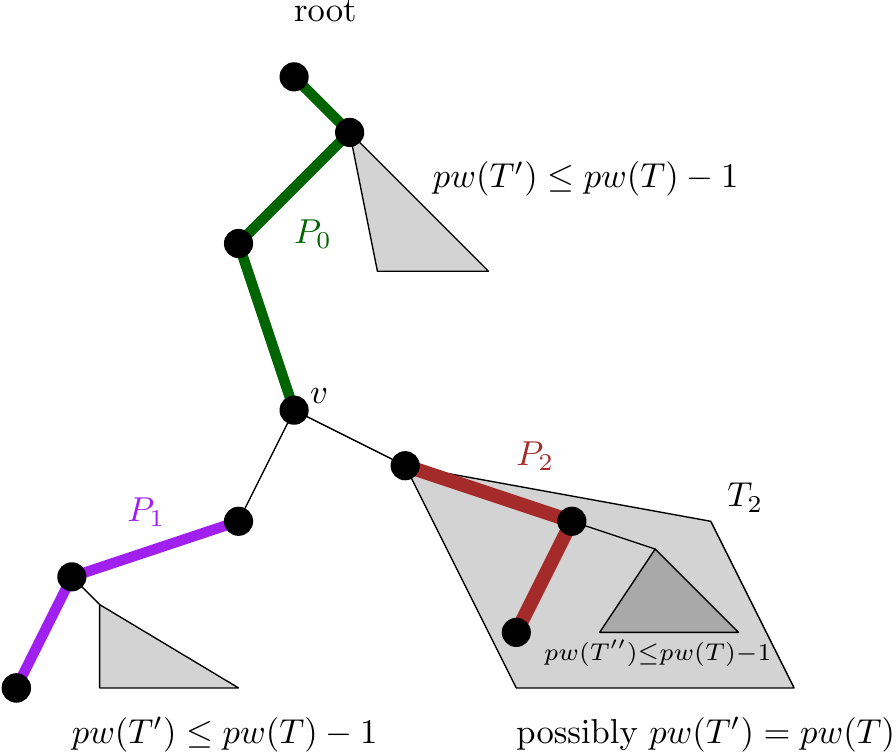}
\hspace{\fill}
\caption{The main path $P_1-v-P_2$ can be used to show $rpw(T_2)\leq 2pw(T)$
and therefore $rpw(T)\leq 2pw(T)-1$.}
\label{fig:tree_paths}
\end{figure}

\subsection{Heavy-path decompositions:}

The {\em heavy-path decomposition}, first introduced by Sleator and
Tarjan \cite{ST83}, is a method of splitting a tree into
paths such that any root-to-leaf path encounters $O(\log n)$
of these paths.  Let the {\em size-heaviest} child of the root be the child
whose subtree contains the most nodes (breaking ties arbitrarily).
The {\em heaviest path} is obtained by going from the root to a leaf
by always going to the size-heaviest child.  If we remove the heaviest path
and recurse in the children, then after some number of recursions the
remaining tree is empty; this number of recursions is called the 
{\em heaviest-path depth} (and denoted $hpd(T)$).   Formally,
$$
hpd(T) = \left\{
\begin{minipage}{50mm}
$
\begin{array}{ll}
1 & \mbox{if $T$ is a single node} \\
\max_{c} \left\{ hpd(T_c) + \chi_{c\neq c_h}\right\} & \mbox{otherwise}
\end{array}
$
\end{minipage}
\right.
$$
where the maximum  is taken over all children $c$ of the root, and
$c_h$ is the size-heaviest child.  
Note that the recursive formula is very similar to, but more restrictive,
than the one in Definition~\ref{def:rpw}; by induction one easily shows
that $rpw(T)\leq hpd(T)$ for all rooted trees $T$.
%
This is far from tight for some trees.

\begin{lemma}
There exists an infinite number of binary trees $T$ with $rpw(T)=2$
and $hpd(T)\in \Omega(\log n)$.
\end{lemma}
\begin{proof}
Let $T_1$ be a single node.  For $i>1$, let $T_i$ consist of a root
with left subtree $T_{i-1}$ and right subtree a rooted
path of length $|T_{i-1}|+1$.  Clearly $rpw(T_i)=2$, using as path
for Observation~\ref{obs:rpw} the one obtained by always going left,
since the right subtrees are rooted
paths and hence have rooted pathwidth 1.
But the right child is the size-heaviest
child, and therefore $hpd(T_i)= 1+hpd(T_{i-1}) = i$.
Since $|T_i|=2|T_{i-1}|+2 = \frac{3}{2}2^i - 2$, the
result follows.
\end{proof}

The algorithm of Crescenzi et al.~\cite{CDP92}, which inspired
our Lemma~\ref{lem:upper}, works by using the size-heaviest
child as $c_1$, i.e., as the child to be drawn using the full width.
For the above tree, their algorithm hence would use width $\Theta(\log n)$,
whereas our variation that uses the rpw-heaviest child as $c_1$ achieves width 2.

\section{Finding right-corner-$W$-witnesses}

Algorithm~\ref{algo:testRight} gives the algorithm to find right-corner-$W$-witnesses.
We also state the lemmas that show its correctness; their proofs mirror
the ones of Lemma~\ref{lem:testSuccess} and \ref{lem:testFailure} and are left
to the reader.

\begin{algorithm}[ht]
\begin{tabbing}
\hspace*{2cm}\=\hspace*{2cm}\=\kill
// $T$ is a tree with children $c_1,\dots,c_d$, $d\geq 1$, $W\geq 1$ \\
Let $i$ be the minimal index such that $R(T_{c_i})\geq W$ \\
\IF{no such $i$ exists}  \RETURN ``success'' \\
\IF{$R(T_{c_i})>W$} \RETURN ``failure'' \\
Now $c_i$ is the leftmost child with $R(T_{c_i})=W$. \\
Initialize $\sigma(W)$ to be $i$, $w$ to be $W$ and increase $i$ \\
\LOOP \\
\>\WHILE{$i<=d$ \AND $R(T_{c_i})\leq w-2$} increase $i$ \\
\>\IF{$i==d+1$} set $W':=w$ and \RETURN ``success'' \\
\>\IF{$R(T_{c_i})\geq w$} set $W':=w$ and \RETURN ``failure'' \\
\>Now $c_i$ is a child with $R(T_{c_i})=w-1$ and $i>\sigma(w)>\dots>\sigma(W)$\\
\>Set $\sigma(w-1)$ to be $i$, decrease $w$ and increase $i$. \\
\ENDLOOP 
\end{tabbing}
\vspace*{-3mm}
\caption{{\sc TestRight}$(T,W)$}
\label{algo:testRight}
\end{algorithm}

\begin{lemma}
\label{lem:testSuccessR}
Assume algorithm {\sc TestRight} returns with ``success''.  Then $T$ has
a right-corner-$W$-witness.
\end{lemma}

\begin{lemma}
\label{lem:testFailureR}
Assume algorithm {\sc TestRight} returns with ``failure''.  Then $T$ has
no rank-$W$-witness with $v=d$.
\end{lemma}

\section{Straight-line drawings}
\label{ap:straight}

Now we give the proof of Theorem~\ref{thm:quaternary},
which states that the tree $T$ in Figure~\ref{fig:quintinary}(a)
needs strictly more width in a straight-line order-preserving 
drawing than in a poly-line drawing.

\begin{proof}
The figure shows a poly-line drawing of $T$ with width~2.
Observe that $u_3$ has rank 2 since $u_3$ has two children
of rank 1.  Therefore the rank-sequence of the children of $u_2$
contains $1,1,2$ as a subsequence. Applying algorithm {\sc TestLeft}(2)
shows that therefore $u_2$ has no left-corner-$2$-witness. 
Likewise $u_1$ has no left-corner-$2$-witness since $u_2$ has
rank 2 and so the ranks of children of $u_1$ include $1,1,2$ as a subsequence.
By Lemma~\ref{lem:equivalent} therefore $u_i$ (for $i=1,2$)
does not have a rank-2-witness with $X=1$.  By
Lemma~\ref{lem:draw2rank} therefore no drawing of $T_{u_i}$ of
width 2 has $u_i$ in column 1.

Fix an arbitrary upward order-preserving drawing $\Gamma$ 
of $T$ of width 2.  For $i=1,2$, the induced drawing of $T_{u_i}$ has also 
width 2, and by the above $u_i$ must be drawn in column 2.  
This drawing cannot be straight-line, else
$\overline{u_1 u_2}$ would be vertical, making it impossible to draw
the rightmost child of $u_1$ while preserving the order.  So
any such drawing of width 2 contains bends.
\end{proof}

If we replace any leaf in $T$ with a subtree that
requires width $W-1$ (e.g.~a binary tree of height $W-1$), then
much the same proof shows that this tree has a poly-line drawing
of width $W$, but no straight-line drawing of width $W$.

\medskip

Now we give the proof of Theorem~\ref{thm:not_corner},
which states that in an optimum-width straight-line order-preserving drawing 
 of the tree $T$ in Figure~\ref{fig:quintinary}(b),
the root cannot be in the middle.

\begin{proof}
The root of tree $T$ has four children $c_1,c_2,c_3,c_4$.
$T_{c_1}$ and $T_{c_4}$ are single nodes.  $T_{c_2}$ is a symmetric
version of the tree in Figure~\ref{fig:quintinary}(a), hence it
requires width 2, and in any width-2 drawing the root must be
in the top-left corner.  $T_{c_3}$ is the tree in Figure~\ref{fig:quintinary}(a)
with leaves replaced by binary trees of height 2; hence it requires width 3,
and in any width-3 drawing the root must be in the top-right corner.

Fig.~\ref{fig:orderedNotStraight}(right) shows a
straight-line drawing with width 3.  Presume we had a straight-line
drawing of $T$ of width 3 where the root $u_r$ is in the top left corner.  
Since $T_{c_3}$ requires width 3, it contains a point $p_3$ in column 1. 
The poly-line from
$u_r$ to $p_3$ blocks $T_{c_2}$ from using column 3, so $T_{c_2}$ must
be drawn with width 2 and hence $c_2$ is in column 1.  Now the 
straight-line segment $\overline{u_r c_2}$ is vertical and $c_1$ cannot
be drawn.  Likewise, if $u_r$ is in the top right corner, then (since
$c_3$ must be in column 3) the straight-line segment $\overline{u_r c_3}$
prevents $c_4$ from being drawn.  Thus the root cannot be in a corner.
\end{proof}

\section{Bounds on the rank}

The algorithm implicit in Lemma~\ref{lem:corner2draw}
draws trees upward and order-preserving with optimal
width, but how big is this width?  We know $R(T)\in O(\log n)$
from Chan's work \cite{Chan02}.  The complete binary
tree has $R(T)\geq \log(n+1)$, so asymptotically this is tight.  We now 
show that the lower bound is in fact tight up to a small additive constant.  

\begin{lemma}
\label{lem:rank_bound}
Any $n$-node tree $T$ has $R(T)\leq \log n + 1$.
\end{lemma}
\begin{proof}
Let $N(W)$ be
the minimum number of nodes in a tree that has rank $W$.
We aim to show that $N(W)\geq 2^{W-1}$; this proves the claim.

Clearly $N(1)\geq 1 = 2^0$, so the claim holds for $W=1$.
Assume it holds for all values up to $W$, and let $T$
be a node-minimal tree that has rank $W+1$.  No child of
$T$ can have rank $W+1$ by minimality of $T$, so the ranks
of the children belong to $\{1,\dots,W\}$.
Let $W^*\leq W$ be the largest value such that root does {\em not} have
exactly one child with rank $W^*$.  (Hence there might be zero or at
least 2 children with rank $W^*$.) 

Assume first that $T$ has no child of rank $W^*$, and exactly one
child each of rank $W^*+1,\dots,W$.  Applying algorithm
{\sc TestLeft}($W$), one sees that it will return
with success at some $W'\geq W^*+1$, 
so $R(T)\leq W$,
a contradiction.  So there must be at least two children of rank $W^*$.   
The subtree of the child with rank $i$ has at least $N(i)$ nodes,
so 
$ N(W+1)=|T| \geq N(W)+N(W-1)+\dots + N(W^*+1) + 2\cdot N(W^*),$
and by induction therefore
$N(W+1) \geq 2^{W-1} + 2^{W-2} + \dots + 2^{W^*} + 2\cdot 2^{W^*-1} = 2^W$
as desired.
\end{proof}

We note here that the bound is not tight (for example, we can
add a `+1' in the final inequality, since we did not count the root).
By distinguishing a large number of cases we have been able to show that
$N(W)\geq \frac{3}{2}2^{W-1}$.
We suspect that in fact $N(W)\geq 2^W-1$,
but the enormous work to prove this does not seem worth the minor
improvement in the bound on $R(T)$.

So both the rooted pathwidth and the rank are $\log n + O(1)$ in the worst case.
One may wonder whether perhaps they are within a constant of each other for all
trees?  This is not the case.

\begin{theorem}
For any $i\geq 1$, there exists a tree $T_i$ with degree 5
that has a planar upward drawing of width $i$ (hence rooted pathwidth at most $i$), 
but its rank is $2i-1$, and so any planar order-preserving upward drawing requires width at least $2i-1$.
\end{theorem}
\begin{proof}
$T_1$ is a single node, which can be drawn with width 1
and requires width at least $1=2\cdot 1-1$.

For $i\geq 2$, tree
$T_i$ consists of a node with degree 5 for which children $c_1,c_2,c_4,c_5$
are roots of $T_{i-1}$.  Child $c_3$ has two children,  
each of which is the root of $T_{i-1}$.    See Fig.~\ref{fig:quintinary}(c)
and (d),
which also illustrates how to obtain an unordered drawing of $T_i$ with 
width $i$.

We show that $R(T_i)\geq 2i-1$.  Clearly this holds for $T_1$, so 
assume we know that $R(T_{i-1})\geq 2i-3$.  Since $c_3$
has two children with rank $2i-3$, $T_{c_3}$ has rank at least $2i-2$.
Therefore the rank-sequence of children contains $2i-3,2i-3,2i-2$ from
left to right.  Applying {\sc TestLeft}($2i-2$) therefore will result
in failure, so $T_i$ has no
left-corner-$(2i-2)$-witness.  Likewise the rank-sequence
$2i-2,2i-3,2i-3$ means that $T_i$ has no right-corner-$(2i-2)$-witness.
By Lemma~\ref{lem:rank2corner} therefore $T_i$ has no
rank-$(2i-2)$-witness and $R(T_i)\geq 2i-1$ as desired.
\end{proof}

\end{appendix}
\end{document}